\documentclass[11pt]{article} \usepackage{pgfplots}
\pgfplotsset{width=7cm,compat=1.9}
\usepackage{authblk}
\usepackage{amsmath,amsthm,amssymb, tabu} \usepackage{booktabs} \usepackage{float} \usepackage[margin=1in]{geometry} \usepackage{graphicx} \usepackage{mathtools}

\usepackage{mdframed} \usepackage{microtype} \usepackage{mathpazo} \usepackage{tikz} \usepackage[scr=rsfso]{mathalfa}

\mdfdefinestyle{figstyle}{   linecolor=black!7,   backgroundcolor=black!7,   innertopmargin=10pt,   innerleftmargin=25pt,   innerrightmargin=25pt,   innerbottommargin=10pt }

\renewcommand{\cal}[1]{\mathcal{#1}}

\newcommand{\microspace}{\mspace{.5mu}} \newcommand{\ket}[1]{\ensuremath{\lvert\microspace #1
    \microspace\rangle}} 
\newcommand{\paren}[1]{(#1)}

\newcommand{\class}[1]{\mathsf{#1}} \newcommand{\MIP}{\class{MIP}} \newcommand\MIPco{\ensuremath{\class{MIP}^\class{co}}} \newcommand{\RE}{\class{RE}} \newcommand{\MIPzero}{\mathsf{MIP}^*_0}
\newcommand{\MIPstar}{\mathsf{MIP}^*}
\newcommand{\PItwo}{\Pi_2^0}
\newcommand{\coRE}{\mathsf{coRE}}

\DeclareMathOperator{\HALT}{\mathrm{Halt}}
\DeclareMathOperator{\COMPRESS}{\mathrm{Compress}}

\newcommand{\C}{\mathbb{C}}
\newcommand{\N}{\mathbb{N}}

\let\epsilon=\varepsilon

\newtheorem{theorem}{Theorem}  \newtheorem{proposition}[theorem]{Proposition} \newtheorem{corollary}[theorem]{Corollary} \newtheorem{definition}[theorem]{Definition}   \newtheorem{claim}[theorem]{Claim}   
\newenvironment{gamespec}{
  \begin{mdframed}[style=figstyle]}{
  \end{mdframed}}

\begin{document}

\title{On the complexity of zero gap $\MIP^*$}

\author[]{Hamoon Mousavi}
\author[]{Seyed Sajjad Nezhadi}
\author[]{Henry Yuen}

\affil[]{Department of Computer Science, University of Toronto, Toronto, Canada. \authorcr
  \{\tt hmousavi@cs.toronto.edu, sajjad.nezhadi@mail.utoronto.ca, hyuen@cs.toronto.edu\}}

\date{}

\maketitle
\begin{abstract}
The class $\mathsf{MIP}^*$ is the set of languages decidable by multiprover interactive proofs with quantum entangled provers. It was recently shown by Ji, Natarajan, Vidick, Wright and Yuen that $\mathsf{MIP}^*$ is equal to $\mathsf{RE}$, the set of recursively enumerable languages. In particular this shows that the complexity of approximating the quantum value of a non-local game $G$ is equivalent to the complexity of the Halting problem. 

In this paper we investigate the complexity of deciding whether the quantum value of a non-local game $G$ is \emph{exactly} $1$. This problem corresponds to a complexity class that we call \emph{zero gap} $\mathsf{MIP}^*$, denoted by $\mathsf{MIP}^*_0$, where there is no promise gap between the verifier's acceptance probabilities in the YES and NO cases. We prove that $\mathsf{MIP}^*_0$ extends beyond the first level of the arithmetical hierarchy (which includes $\mathsf{RE}$ and its complement $\mathsf{coRE}$), and in fact is equal to $\Pi_2^0$, the class of languages that can be decided by quantified formulas of the form $\forall y \, \exists z \, R(x,y,z)$.

Combined with the previously known result that $\mathsf{MIP}^{co}_0$ (the \emph{commuting operator} variant of $\mathsf{MIP}^*_0$) is equal to $\mathsf{coRE}$,  our result further highlights the fascinating connection between various models of quantum multiprover interactive proofs and different classes in computability theory.
\end{abstract}
 \section{Introduction}

A two-player \emph{non-local game} is played between a verifier and two cooperating players named Alice and Bob who cannot communicate with each other once the game starts. During the game, the verifier samples a pair of questions $(x,y)$ from a joint distribution $\mu$, sends $x$ to Alice and $y$ to Bob, who respond with answers $a$ and $b$ respectively. The verifier accepts if and only if $D(x,y,a,b) = 1$ for some predicate $D$. The \emph{quantum value} of a non-local game $G$, denoted by $\omega_q(G)$, is defined to be the supremum of the verifier's acceptance probability over all possible finite dimensional quantum strategies of Alice and Bob for the game $G$. 
            
What is the complexity of computing the quantum value of non-local games? In~\cite{slofstra2019set}, Slofstra proved that the problem of determining whether a given game $G$ has $\omega_q(G) = 1$ is \emph{undecidable}. Recently, it was shown that \emph{approximating} $\omega_q(G)$ up to any additive constant is also an uncomputable problem~\cite{ji2020mip}. In particular, there is a computable reduction from Turing machines $M$ to non-local games $G_M$ such that if $M$ halts (when run on an empty input), then $\omega_q(G_M) = 1$, and otherwise $\omega_q(G_M) \leq \frac{1}{2}$. Since determining whether a given Turing machine halts (i.e. the Halting problem) is undecidable, so is the problem of determining whether the quantum value of a non-local game is $1$ or at most $\frac{1}{2}$. 

Conversely, one can reduce the problem of approximating the quantum value of non-local games to the Halting problem; there is an algorithm that for every non-local game $G$ exhaustively searches over finite-dimensional strategies of increasing dimension to find one that succeeds with probability close to $1$ (above $0.99$, say). If $\omega_q(G) = 1$ then the algorithm is guaranteed to find such a strategy; otherwise if $\omega_q(G) \leq 1/2$ the algorithm will run forever. In complexity-theoretic terms, this shows that the class $\MIP^*$, the set of languages decidable by multiprover interactive proofs with quantum provers, is equal to $\RE$, the set of recursively enumerable languages (i.e. the class for which the Halting problem is complete). 

In this paper, we return to the problem originally investigated by Slofstra~\cite{slofstra2019set}: what is the complexity of deciding if $\omega_q(G)$ is \emph{exactly} equal to $1$ for nonlocal games $G$? This corresponds to the complexity class that we call \emph{zero gap} $\MIP^*$, denoted by $\MIP^*_0$. In this model of interactive proofs, in the YES case (i.e. $x \in L$), there is a sequence of finite-dimensional prover strategies that cause the verifier to accept with probability approaching $1$. In the NO case (i.e. $x \notin L$), all finite-dimensional prover strategies are rejected with positive probability -- but could be arbitrarily close to $0$. 
It is easy to see that $\MIP^* \subseteq \MIP^*_0$ and thus $\MIP^*_0$ contains undecidable languages. Furthermore, we know that $\MIP^*_0$ \emph{cannot} be equal to $\MIP^*$; the results of~\cite{slofstra2019set,Compress2} imply that $\coRE$, the complement of $\RE$, is also contained in $\MIP^*_0$. Since $\RE \neq \coRE$, this implies that $\MIP^*_0$ strictly contains $\MIP^* = \RE$.

What problems can be reduced to the task of exactly computing the quantum value of non-local games, rather than ``just'' approximating it? We characterize the class $\MIP^*_0$ by showing that it is equal to $\PItwo$, a class that belongs to the \emph{arithmetical hierarchy} from computability theory. The arithmetical hierarchy is defined by classes of languages decidable via formulas with alternating quantifiers. For example, the class $\RE$ is equal to the class $\Sigma_1^0$, which is the set of languages $L$ of the form $\{ x : \exists y\ldotp R(x,y) = 1\}$ for some decidable predicate $R$. The class $\coRE$ is equal to $\Pi^0_1$, the set of languages of the form $\{ x : \forall y \ldotp R(x,y) = 1\}$. The class $\PItwo$ is the set of languages $L$ of the form $\{ x : \forall y\ldotp \exists z\ldotp R(x,y,z) = 1 \}$. 

An equivalent definition of the class $\PItwo$ is that it is the set of languages $L$ such that there is a Turing machine $A$ that has \emph{oracle access} to the Halting problem, and $x \notin L$ if and only if $A(x) = 1$. It is known that $\PItwo$ strictly contains $\Sigma_1^0 = \RE$. This shows that $\MIP^*_0$ contains problems that are \emph{harder} (in a computability sense) than the Halting problem.

We specifically show that there exists a computable reduction from $\PItwo$ languages to the problem of deciding whether a \emph{three-player} non-local game $G$ has quantum value $1$. It is likely that a similar reduction holds for two-player non-local games but we leave this for future work. We also show that the problem of deciding if a non-local game has quantum value $1$ can be reduced to a $\PItwo$ language, thus establishing the equality $\MIP^*_0 = \PItwo$. 

This paper, combined with the results of~\cite{ji2020mip} and~\cite{slofstra2019set}, paints a fascinating landscape about the complexity of quantum multiprover interactive proofs, in which there are four different complexity classes to consider. The first two are $\MIP^*$ and $\MIP^*_0$, which we defined already. The second two are $\MIPco$ and its zero-gap variant $\MIPco_0$. The class $\MIPco$ stands for languages that are decidable by quantum multiprover interactive proofs in the \emph{commuting operator} model: here, the provers are allowed to use infinite-dimensional quantum strategies, and the measurement operators of Alice only need to commute with those of Bob (rather than be in tensor product). 

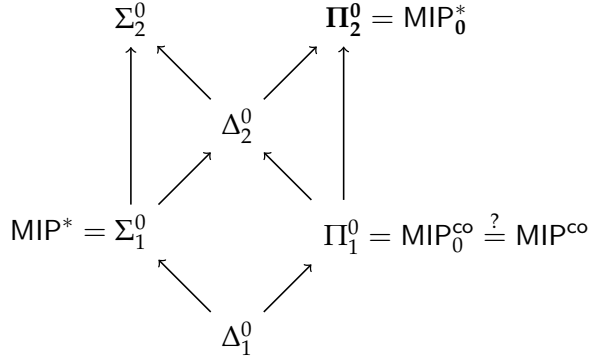
\begin{figure}[H]\label{fig:computability landscape}
\usetikzlibrary{arrows}
\begin{center}
\hspace*{1.6cm} \begin{tikzpicture}[->, node distance=2cm, semithick]
 \node (P) {$\Delta_1^0$};
  \node (Sigma1) [above left of=P]       {$\MIP^* = \Sigma_1^0$ \hspace*{1.3cm}}; 
 \node (Pi1)    [above right of=P]      { \hspace*{2.9cm} $\Pi_1^0 = \MIPco_0 \stackrel{?}{=} \MIPco$};
 \node (Delta2) [above left of=Pi1]     {$\Delta_2^0$};
 \node (Sigma2) [above left of=Delta2]  {$\Sigma_2^0$};
 \node (Pi2)    [above right of=Delta2] {\hspace*{1.3cm} $\mathbf{\PItwo = \MIPzero}$};
 \draw (P)      -> (Sigma1);
 \draw (P)      -> (Pi1);
 \draw (Sigma1) -> (Sigma2);
 \draw (Sigma1) -> (Delta2);
 \draw (Pi1)    -> (Pi2);
 \draw (Pi1)    -> (Delta2);
 \draw (Delta2) -> (Sigma2);
 \draw (Delta2) -> (Pi2);
 \end{tikzpicture}
\caption{The computability landscape of quantum multiprover interactive proofs. Arrows denote inclusion. The set $\Delta_1^0$ denotes the set of all decidable languages. The set $\Sigma_1^0$ denotes the recursively enumerable languages, and $\Pi_1^0$ denotes the set of co-recursively enumerable languages. It is known that $\MIPco \subseteq \MIPco_0$, but unknown whether they are equal.}
\end{center}

\end{figure}

One of the consequences of the fact that $\MIPstar = \RE$ is that $\MIPco \neq \MIPstar$. This is because $\MIPco \subseteq \coRE$, due to the fact that the commuting operator value of a non-local game can be upper-bounded using a convergent sequence of semidefinite programs~\cite{NPA,doherty2008quantum}. It is also the case that $\MIPco_0 \subseteq \coRE$, and in fact equality holds due to~\cite{slofstra2019set,coudron2019complexity}. It remains an open question to determine if $\MIPco = \MIPco_0 = \coRE$. 

There are a number of curious and counter-intuitive aspects about this landscape of complexity for non-local games. First, if $\MIPco = \coRE$, then there would be a pleasing symmetry in that $\MIPstar = \RE$ and $\MIPco = \coRE$ (even though the ``co'' refer to different things on each side of the equation!). On the other hand, we have that $\MIPzero = \Pi^0_2$ and $\MIPco_0 = \coRE$, meaning that -- in the zero gap setting -- there are \emph{more} languages that can be verified with provers using (a limit of) finite-dimensional strategies than can be decided with provers using infinite-dimensional commuting operator strategies! Of course, in the setting of interactive proofs, giving provers access to more resources can change the complexity of the interactive proof model in unexpected ways.

\subsection{Proof overview}

We prove the lower bound $\PItwo \subseteq \MIPzero$ by combining two components: first we leverage the result of~\cite{ji2020mip} that $\MIPstar = \RE$ as a black box, which implies that there is a quantum multiprover interactive proof for the Halting problem. Next, we use a \emph{compression theorem} for quantum multiprover interactive proofs that was proved in~\cite{Compress2}. A compression theorem, roughly speaking, states that given a verifier $V$ for a quantum multiprover interactive protocol (which can be modeled as a Turing machine with tapes to receive/send messages to the provers), one can compute a much more time-efficient verifier $V'$ whose quantum value is related in some predictable way to the quantum value of $V$. Several recent results about the complexity of non-local games crucially rely on proving compression theorems with various properties~\cite{ji2017compression,Compress2,natarajan2019neexp,ji2020mip}. 

In more detail, the compression theorem of~\cite{Compress2} (which in turn is a refinement of the compression theorem of~\cite{ji2017compression}) states that given a description of a verifier $V$, one can compute a description of a three-player\footnote{The results of~\cite{Compress2} are stated for games with $15$ players, but can be improved to hold for $3$-player games by using a different error correcting code in the construction.} non-local game $G_V$ (which is a multiprover protocol with only one round of interaction) whose properties are as follows:
\begin{enumerate}
    \item The time complexity of the verifier in $G_V$ is \emph{polylogarithmic} in the time complexity of $V$.
    \item The quantum value of the protocol executed by $V$ is related to the quantum value of $G_V$ in the following manner:
    \[
        \omega_q(G_V)  \geq \frac{1}{2} + \frac{1}{2} \omega_q(V)
    \]
    and furthermore if $\omega_q(V) < 1$ then $\omega_q(G_V) < 1$.
\end{enumerate}
The utilization of the compression theorem of~\cite{Compress2} is the reason why the main result of this paper holds for three-player non-local games, rather than two.

We call this compression theorem a ``zero gap'' compression theorem, because it does not preserve any promise gap on the value of the input verifier $V$: if the value of $V$ is promised to be either $1$ or $1/2$, then $G_V$ is only guaranteed to have value either $1$ or $3/4$. If we iterate this compression procedure, then we get a promise gap that goes to zero. In contrast, the compression theorem used to prove $\MIPstar = \RE$ \emph{is} gap-preserving.

The zero gap compression theorem was used to prove that $\coRE \subseteq \MIPzero$ in~\cite{Compress2}. At a high level, this is shown by constructing a verifier that recursively calls the zero gap compression procedure on itself. In this paper, we follow this approach, except we also embed an $\MIPstar$ protocol for $\RE$ inside the verifier that is recursively calling the zero gap compression procedure; this composition of protocols allows the verifier to verify languages in $\PItwo$.

\subsection{Further remarks}
 
\paragraph{$\MIP^* = \RE$ is equivalent to gap-preserving compression.} As mentioned, the key to proving $\MIPstar = \RE$~\cite{ji2020mip} was establishing a gap-preserving compression theorem for non-local games, albeit for a special case of non-local games satisfying a so-called ``normal form'' property. In Section~\ref{sec:gap}, we present a relatively simple -- but in our opinion quite interesting -- observation that $\MIPstar = \RE$ is in some sense, \emph{equivalent} to a gap-preserving compression theorem.

\paragraph{A proof of $\MIP^*_0 = \Pi_2^0$ under weaker assumptions?} One might wonder if there might be an elementary way of proving that $\MIP^*_0 = \Pi_2^0$, \emph{without} relying on the statement that $\MIP^* = \RE$. For example, the results of~\cite{slofstra2019set,Compress2} show that $\coRE \subseteq \MIP^*_0$ and furthermore~\cite{slofstra2019set} shows that $\coRE = \MIP^{co}_0$. These previous ``zero-gap results'' do not appear to have the same mathematical consequences as $\MIP^* = \RE$ (e.g. yielding a negative answer to Connes' embedding problem if $\RE \subseteq \MIP^*(2)$, the two-player variant of $\MIP^*$), which suggests the intuition that characterizing the complexity of \emph{exactly} computing the quantum (or commuting operator) value of nonlocal games may be fundamentally easier than characterizing the complexity of \emph{approximating} it. 

This intuition is not entirely correct: the ``zero-gap'' statement $\MIP^*_0 = \Pi_2^0$ is already enough to yield a negative answer to Tsirelson's problem: there exists a $k$ where $k$-partite commuting operator correlations cannot be approximated by finite dimensional correlations. Put another way, if Tsirelson's problem has a positive answer, then the commuting operator and quantum values of games are always equal, and then $\MIP^*_0 = \MIP^{co}_0 = \coRE$. However, $\Pi_2^0$ strictly contains $\coRE$ -- thus Tsirelson's problem has a negative answer. Furthermore, Tsirelson's problem for $k = 2$ is known to be equivalent to Connes' embedding problem~\cite{fritz2012tsirelson,junge2011connes,Ozawa}.

This suggests that our characterization of the class $\MIP^*_0$ must necessarily involve a nontrivial tool such as $\MIP^* = \RE$.

\subsection{Open problems}

We list some open problems.
\begin{enumerate}
    \item Just as the complexity statement $\MIPstar = \RE$ has consequences for questions in pure mathematics (such as the Connes' embedding problem), does the equality $\MIPzero = \PItwo$ have any implications for operator algebras? We believe there may be a connection to model-theoretic approaches to the Connes' embedding problem (see, e.g.,~\cite{goldbring2013computability,goldbring2017enforceable}).
    
    \item What is the complexity of $\MIPco$? Is it equal to $\coRE$?
    
    \item Can the reduction from $\PItwo$ languages to the problem of deciding whether $\omega_q(G) = 1$ be improved to hold for two-player games $G$?
    
    \item We showed that, essentially, $\MIPstar = \RE$ implies a gap-preserving compression theorem. Can one show that it also implies in a black-box fashion, a zero gap compression theorem, of the same kind as proved in~\cite{Compress2}? This then proves that $\MIPstar = \RE$ directly implies $\MIPzero=\PItwo$.
    
    \item Does $\MIPzero = \PItwo$ imply $\MIPstar = \RE$ in a ``black-box'' fashion?
\end{enumerate}

\subsection*{Acknowledgments} 

We thank Matt Coudron, Thomas Vidick, and especially William Slofstra for numerous helpful discussions. We also thank the reviewers of ICALP 2020 for suggestions to improve the presentation. HY was supported by NSERC Discovery Grant 2019-06636. HM was supported by the Ontario Graduate Scholarship (OGS).

 \section{Preliminaries}

We write $\N$ to denote the natural numbers $\{1,2,3,\ldots\}$. All logarithms are base $2$. For a string $x \in \{0,1\}^*$ let $|x|$ denote the length of $x$. We let $$\log^*(n) = 
\begin{cases}
0,   &n \leq 1 \\
1 + \log^*(\log(n)), &n > 1 
\end{cases}$$ 
denote the iterated logarithm function.

\subsection{Turing machines and the arithmetical hierarchy}

A total Turing machine is one that halts on every input. Fix a string encoding of Turing machines, and for a Turing machine $M$, let $|M|$ denote the length of the encoding of $M$.

\begin{proposition}[Universal Turing machine]
\label{prop:universal-tm}
There exists a universal constant $C >0$ and a universal Turing machine $\mathscr{U}$ that, given an input pair $(M,x)$ where $M$ is an encoding of a Turing machine, computes $M(x)$ in time $C\max(|M|, \mathsf{TIME}(M,x))^2$, where $\mathsf{TIME}(M,x)$ is the number of steps taken by $M$ on input $x$ before it halts. 
\end{proposition}

\begin{definition}
The $i$-th level of the \emph{arithmetical hierarchy} contains $3$ classes $\Sigma^0_i$, $\Pi^0_i$, and $\Delta^0_i$. The class $\Sigma^0_i$ is the set of languages defined as $$L = \{x \in \{0,1\}^\ast : \exists y_1 \forall y_2 \exists y_3 \, \cdots  \,Q\, y_i \, R(x,y_1, \cdots,y_i) = 1\}$$ for some total Turing machine $R$, where $Q$ is the $\forall$ quantifier when $i$ is even and otherwise is the $\exists$ quantifier. The class $\Pi^0_i$ is the complement of $\Sigma_i^0$, and $\Delta^0_i = \Sigma^0_i \cap \Pi^0_i$.
\end{definition}

In particular the first level of the arithmetical hierarchy corresponds to the classes $\Sigma^0_1 = \RE$, $\Pi^0_1 = \coRE$, and $\Delta^0_1$ the set of decidable languages $\RE \cap \coRE$.

\subsection{Interactive verifiers}

In this section, we model multiprover interactive protocols, which is specified by a \emph{verifier} $V$, as a randomized algorithm. In the protocol, the verifier $V$ interacts with multiple provers, and at the end of the protocol the verifier outputs a bit indicating whether to accept or reject. A verifier can be identified with the interactive protocol it executes, and vice versa.

In more detail, define a \emph{$k$-input, $r$-prover verifier $V$} to be a randomized interactive Turing machine that has $k$ designated input tapes, $r$ communication tapes, a single workspace tape, and a single output tape. An interaction with $r$ provers is executed in the following way: the Turing machine $V$ alternates between computation and communication phases; in the computation phase, the Turing machine behaves like a normal Turing machine with $k+r+2$ tapes, and it may halt and indicate accept or reject on the output tape. It can also pause its computation and go into a communication phase, in which case the contents of each of $i$-th communication tape is read by the $i$-th prover, who then edits the $i$-th communication tape with its answer. After all the provers have finished with their responses, the next computation phase resumes. This is the standard way of modeling interactive Turing machines~\cite{ben1988multi}. In this formulation, a non-local game is simply specified by a $0$-input, $2$-prover verifier $V$ that has only one communication phase.

Given a $k$-input, $r$-prover verifier $V$, define its \emph{time complexity} with respect to a $k$-tuple of inputs $(x_1,\ldots,x_k)$ to be the maximum number of time steps taken by the verifier $V$ when it is initialized with $(x_1,\ldots,x_k)$ on its $k$ input tapes, over all possible responses of the $r$-provers, before it halts. We denote this by $\mathsf{TIME}(V(x_1,\ldots,x_k))$. 

We now define, in a somewhat informal level, \emph{finite-dimensional prover strategies} (or simply a \emph{strategy}) $\cal{S}$ for the interaction specified by a $k$-input, $r$-prover verifier $V$. This is a specification of the following data: 
\begin{enumerate}
    \item Local dimension $d \in \N$, 
    \item A state $\ket{\psi} \in (\C^d)^{\otimes r}$, and
    \item For every prover $i$, for every round $t \in \N$, for every string $\pi \in \{0,1\}^*$, a POVM $\{M_{i,t,\pi}^a \}_a$ acting on $\C^d$.
\end{enumerate}

Given a verifier $V$, a $k$-tuple $(x_1,\ldots,x_k)$, and a prover strategy $\cal{S}$ for $V$, the interaction proceeds as follows: at the beginning of the protocol, the provers share the state $\ket{\psi}$, and the verifier's input tapes are initialized to $(x_1,\ldots,x_k)$. At round $t$, the $i$-th prover performs the measurement $\{ M_{i,t,\pi}^a \}_a$ on its local space to obtain an outcome $a$, where $\pi$ is the \emph{history} of all the messages seen by prover $i$ in all previous rounds (including the message from the verifier in the $t$-th round). It then writes outcome $a$ on the $i$-th communication tape of the verifier. Thus at each round the shared state between the provers depend on the outcomes of their measurements, and evolves probabilistically over time. The \emph{value of strategy $\cal{S}$ in the interaction with verifier $V$ on input $(x_1,\ldots,x_k)$} is defined to be the probability that the verifier halts and accepts. We denote this by $\omega_q(V(x_1,\ldots,x_k),\cal{S})$. The \emph{quantum value of verifier $V$ on input $(x_1,\ldots,x_k)$} is defined to be the supremum of $\omega_q(V(x_1,\ldots,x_k),\cal{S})$ over all finite-dimensional strategies $\cal{S}$, which we denote by $\omega_q(V(x_1,\ldots,x_k))$. 

\begin{definition}
Let $m,r \in \N$ and let $0 \leq s \leq c \leq 1$. The class $\MIPstar[m,r,c,s]$ is defined to be the set of languages $L$ for which there exists a verifier $V$ and a polynomial $p(n)$ with the following properties:
\begin{enumerate}
    \item $V$ is a $1$-input, $r$-prover verifier that halts after $m$ communication phases.
    \item For all $x$, $\mathsf{TIME}(V(x)) \leq p(|x|)$. 
    \item If $x \in L$, then $\omega_q(V(x)) \geq c$.
    \item If $x \notin L$, then $\omega_q(V(x)) < s$.
\end{enumerate}
\end{definition}

We define the class $\MIPstar$ to be the union of $\MIPstar[m,r,c,s]$ for all $m,r \in \N$ and $c > s$. We define the class $\MIPzero$ to be the union of $\MIPstar[m,r,1,1]$ over all $m,r \in \N$. In other words, in the YES case (i.e., $x \in L$), there is a sequence of finite-dimensional prover strategies that are accepted with probability approaching $1$. In the NO case (i.e., $x \notin L$), there exists a positive $\epsilon > 0$ (that generally depends on $x$) such that all finite dimensional strategies are rejected with probability at least $\epsilon$.

\subsection{Compression of quantum multiprover interactive protocols}

In this section we formally present the two main ingredients used in our proof: the zero gap compression procedure of~\cite{Compress2}, and the reduction from the Halting problem to the problem of approximating the quantum value of a quantum multiprover interactive protocol. 

First we introduce the definition of $\lambda$-boundedness, which specifies how the time complexity of a verifier is bounded by a polynomial with exponent $\lambda$.

\begin{definition}
Let $\lambda \in \N$. A $(k+1)$-input $r$-prover verifier $V$ is $\lambda$-bounded if for all integers $n\in \N$ and strings $x_1,\ldots,x_k\in \{0,1\}^*$, we have $\mathsf{TIME}(V(n,x_1, ...,x_k)) \leq \lambda(n \cdot |x_1| \cdots |x_k| )^\lambda$. 
\end{definition}

Here, we assume that the first input to a verifier $V$ is an integer $n \in \N$ which intuitively specifies a ``complexity'' parameter. 

\begin{theorem}[Zero-gap compression{~\cite[Theorem 6.1]{Compress2}}] \label{Compression}
 Let $r \geq 3$ be an integer. There exists a universal constant $C_{comp} \in \N$ such that for every $\lambda \in \N$, there exists a Turing machine $\COMPRESS_\lambda$ with the following properties. Given as input a $(k+1)$-input $r$-prover verifier $V$, the Turing machine $\COMPRESS_\lambda$ outputs a $(k+1)$-input $r$-prover verifier $V^\#$ in time $C_{comp} (|V| \cdot \lambda)^{C_{comp}}$ with the following properties: for all $x_1,\ldots,x_k \in \{0,1\}^*$, we have
\begin{enumerate}
    \item if $V$ is $\lambda$-bounded, then $\omega_q(V^\#(n,x_1,...x_k)) \geq \frac{1}{2} + \frac{1}{2} \omega_q(V(2^n,x_1,...x_k)),$
    \item if $V$ is $\lambda$-bounded and $\omega_q(V(2^n,x_1,...x_k)) < 1$, then $\omega_q(V^\#(n,x_1,...x_k)) < 1,$
    \item for all integers $n\in \N$, $x_1,\ldots,x_k\in \{0,1\}^*$, we have $\mathsf{TIME}(V^\#(n,x_1, ...,x_k)) \leq C_{comp} (\lambda \cdot n \cdot |x_1| \cdots |x_k| )^{C_{comp}}$. 
\end{enumerate}
\end{theorem}
The zero-gap compression theorem, as presented here, differs from the one presented in~\cite[Theorem 6.1]{Compress2}. For example, verifiers in~\cite{Compress2} are described using so-called ``Gate Turing Machines'' (GTMs). However, using the same oblivious Turing machine simulation techniques as discussed in the appendix of~\cite{Compress2}, from a verifier $V$ (as defined in this paper), we can obtain a GTM that specifies the same interactive protocol. Another difference, as remarked in the introduction, is that here the compression result applies to protocols with three or more players, whereas it is stated for protocols with $15$ or more players in~\cite{Compress2}. However, the results of~\cite{Compress2} can be adapted to the case of three players by using a $[[3,1,2]]_3$ error detecting code with qutrits (instead of using the $7$-qubit Steane code with qubits)~\cite{cleve1999share}.

\medskip 

Next we present the main result of~\cite{ji2020mip}, which presents a computable reduction from the Halting problem to the problem of approximating the quantum value of a non-local game. 

\begin{theorem}[$\MIP^* = \RE${~\cite{ji2020mip}}]\label{HALT} 
There exists a Turing machine $H$ and a universal constant $C_{\HALT}\in \N$ with the following properties. Given as input a Turing machine $M$, it runs in time $C_{\HALT}|M|^{C_{\HALT}}$ and outputs a $0$-input $2$-prover verifier $V_{\HALT,M}$ such that 
\begin{enumerate}
    \item If $M$ halts on empty tape then $\omega_q(V_{\HALT,M}) = 1$, and otherwise $\omega_q(V_{\HALT,M}) \leq \frac{1}{2}$.
    \item $\mathsf{TIME}(V_{\HALT,M}) \leq C_{\HALT}|M|^{C_{\HALT}}$.
\end{enumerate}
\end{theorem}

 \section{$\MIPzero = \PItwo$}

We start this section by showing the upper bound $\MIPzero \subseteq \Pi_2^0$. 
\begin{theorem} \label{upperbound}
$\MIPzero \subseteq \Pi_2^0$
\end{theorem}
\begin{proof}
Let $L \in \MIP^\ast_0$. There exists a $1$-input $r$-prover verifier $V$ such that $x \in L$ iff $\omega_q(V(x)) = 1$ for all $x \in \{0,1\}^*$. Let $\mathscr{S}_{\epsilon, d}$ be an $\epsilon$-net for the space of strategies of dimension $d$; in particular, for every dimension-$d$ strategy $\cal{S}$ there exists a strategy $\cal{S}' \in \mathscr{S}_{\epsilon,d}$ such that for all verifiers $V$ we have that $|\omega_q(V,\cal{S}) - \omega_q(V,\cal{S}')|\leq \epsilon$ (in other words, the winning probability of the strategies differ by at most $\epsilon$). Because the set of strategies over a finite dimensional Hilbert space of a fixed dimension is a compact set \cite{compact}, we can take $\mathscr{S}_{\epsilon,d}$ to be a finite set. Let $\mathscr{S}_\epsilon = \bigcup_{d \in \N} \mathscr{S}_{\epsilon, d}$, and let $\{ \cal{S}_\epsilon(1),\cal{S}_\epsilon(2),\ldots\}$ be an enumeration of strategies in $\mathscr{S}_\epsilon$.

Consider the following total Turing machine $T$: On input triple $(x,m,n)$ where $x\in \{0,1\}^\ast,m,n \in \N$. It outputs $1$ if and only $\omega_q(V(x),\mathcal{S}_{1/2m}(n)) \geq 1-1/m$. Now it is easy to verify that \[L =  \{x : \forall m\ldotp\ \exists n\ldotp T(x,m,n) = 1\},\]
and therefore $L$ is a $\Pi_2^0$ language. 

To see this, let $x \in L$. Then $\omega_q(V(x)) = 1$, and for any gap (i.e. $\frac{1}{m}$) there exists a strategy $S$ such that $\omega_q(V(x), S) \geq 1 - \frac{1}{2m}$. Choosing $\epsilon = 1/2m$, then there must also exist a strategy $S' \in \mathscr{S}_{1/2m}$ such that $\omega_q(V(x), S') \geq \omega_q(V(x), S) - \frac{1}{2m} \geq 1 - \frac{1}{m}$. Therefore $\forall m\ldotp\ \exists n\ldotp T(x,m,n) = 1$. 

Likewise, if $x \notin L$ then there exists $m \in \N$ for which $\omega_q(V(x)) < 1 - \frac{1}{m}$ and so no strategy can win with probability greater or equal to $1 - \frac{1}{m}$. Therefore $\exists m\ldotp\ \forall n\ldotp T(x,m,n) = 0$.
\end{proof}

Now we prove the reverse inclusion. Fix an $L\in \PItwo$ and let $R$ be a total Turing machine such that $L = \{x\in \{0,1\}^\ast : \forall m\ldotp \exists n\ldotp R(x,m,n)=1 \}$. To prove $L \in \MIPzero$, we construct a $2$-input $3$-prover verifier $V$ that takes as input $m\in\N$ and $x \in \{0,1\}^\ast$, and has the key property that $\omega^\ast(V(m,x))=1$ if and only if $\forall m' \geq \log^*(m) \ldotp \exists n \ldotp R(x,m',n)=1$. Therefore $\omega_q(V(1,x))=1$ if and only if $x \in L$. 

We first give the explicit description of a $3$-input $3$-prover verifier $V'$ below. We then use that to construct $V$. In the description of $V'$, we refer to the Turing machine $R_{x,m}$. For every $x\in\{0,1\}^\ast$ and $m\in\N$, $R_{x,m}$ is the Turing machine that on the empty tape enumerates over $n\in \N$ and accepts if $R(x,m,n)= 1$, otherwise it loops forever.  

\begin{figure}[!htb]
  \centering
  \begin{gamespec}
    Input: $(m,x,W)$ where $m \in \N$, $x \in \{0,1\}^\ast$, $W$ is a $3$-input $3$-prover verifier. \\
    Perform the following steps:
    \begin{enumerate}
        \item Compute $V_{\HALT,R_{x,\log^*(m)}} = H(R_{x,\log^*(m)})$ (where $H$ is from Theorem \ref{HALT}).
        \item Execute the interactive protocol specified by the verifier $V_{\HALT,R_{x,\log^*(m)}}$. If the verifier $V_{\HALT,R_{x,\log^*(m)}}$ rejects then reject, otherwise continue. 
        \item Compute $W^\# = \COMPRESS_\lambda(W)$ (where $\COMPRESS_\lambda$ is from Theorem \ref{Compression}).
        \item Execute the interactive protocol specified by the verifier $W^\#(m,x,W)$ and accept if and only if the verifier $W^\#(m,x,W)$ accepts. 
    \end{enumerate}

  \end{gamespec}
  \caption{Specification of the $3$-input $3$-prover verifier $V'$}
\end{figure}

Now let $V$ be the $2$-input $3$-prover verifier that on the input $(m,x)$ runs $V'(m,x,V')$. Informally, $V(m,x)$ first decides $\exists n \ldotp R(x,\log^*(m),n)=1$ by simulating the verifier in $V_{\HALT,R_{x,\log^*(m)}}$ from Theorem \ref{HALT}. Recall that the existence of the $\MIPstar$ protocol $V_{\HALT,R_{x,\log^*(m)}}$ is due to $\MIPstar=\RE$ and the fact that $\exists n \ldotp R(x,\log^*(m),n)=1$ is an $\RE$ predicate. Now if $R(x,\log^*(m),n)=0$ for all $n$, then $R_{x,\log^*(m)}$ never halts. This in turn implies that $V$ rejects with probability at least $1/2$. Otherwise, if $\exists n\ldotp R(x,\log^*(m),n)=1$, $V$ proceeds to run the compression algorithm to obtain $V'^\# = \COMPRESS_\lambda(V')$. It then executes $V'^\#(m,x,V')$. Informally speaking, due to the compression theorem, the execution of $V'^\#(m,x,V')$ has the same effect as recursively executing $V(2^m,x)$. Now the first duty of the verifier $V(2^m,x)$ is to decide $\exists n \ldotp R(x,1+\log^*(m),n)=1$. So we can apply the above reasoning this time on $V(2^m,x)$ instead of $V(m,x)$. Following this reasoning ad infinitum, we establish that $\omega_q(V(m,x))=1$ if and only if $\forall m' \geq \log^*(m) \ldotp \exists n \ldotp R(x,m',n)=1$. This is made precise in the proof of Theorem \ref{mainThm}.

Note that Theorem \ref{Compression} relates $V^\#(m, x)$ to $V(2^m, x)$. That is the reason $\log^*(m)$ (as opposed to $m$) is appearing in Figure 2. Note that as $m$ increases, $\log^\ast(m)$ ranges over all nonnegative integers.

In order to apply Theorem \ref{Compression} to compress $V$ in step 3, we must ensure that the verifier is $\lambda$-bounded for some $\lambda \in \N$. 
\begin{claim} \label{bounded}
There exists a $\lambda \in \N$ such that $V$ is $\lambda$-bounded.
\end{claim}
\begin{proof}
We bound the running time of $V$ by bounding the running time of each of the steps in its specification. The time to compute the description of $R_{x,\log^*(m)}$, in step 1, is $C\paren{|R| \cdot |x| \cdot m}^C$ for some universal constant $C$. The time to compute the encoding of $V_{\HALT,R_{x,\log^*(m)}}$ is $C_{\HALT}(|R| \cdot |x| \cdot m)^{C_{\HALT}}$. This also bounds the running time of $V_{\HALT,R_{x,\log^*(m)}}$. Therefore the time to simulate $V_{\HALT,R_{x,\log^*(m)}}$ is bounded by $C_{\HALT}^2(|R|\cdot |x| \cdot m)^{2C_{\HALT}}$. The time to simulate ${\COMPRESS_\lambda}(V)$ is $C_{comp}^2 (|V|\cdot \lambda)^{2C_{comp}}$. The time to simulate $V^\#(m,x)$ is bounded by $C_{comp}^2 (\lambda \cdot m \cdot |x|)^{2C_{comp}}$. Therefore the running time of  $V(m,x)$ is bounded above by 
\[2C_{\HALT}^2(|R|\cdot |x|\cdot m)^{2C_{\HALT}} + C\paren{|R|\cdot |x|\cdot m}^C + C_{comp}^2 (|V|\cdot \lambda)^{2C_{comp}} + C_{comp}^2 (\lambda \cdot m \cdot|x|)^{2C_{comp}}.\]

The values $C_{comp}, C_{\HALT}, C,$ and $|R|$ are all constants so we can choose $\lambda \in \N$ sufficiently large so that $\lambda (m \cdot|x|)^\lambda$ is larger then the quantity above and therefore $V$ is $\lambda$-bounded.
\end{proof}

Now that we established that $V$ is $\lambda$-bounded, we can apply Theorem \ref{Compression} to get the main theorem of this paper. 

\begin{theorem} \label{mainThm}
$x \in L$ if and only if $\omega_q(V(1,x)) = 1$
\end{theorem}
\begin{proof}
First suppose $x \in L$. Then $\forall m \ldotp \exists n \ldotp R(x,m,n) = 1$. Since the Turing machine $R_{x,m}$ halts for every $m \in \N$, by Theorem \ref{HALT}, $\omega_q(V_{\HALT,R_{x,m}}) = 1$. Therefore $\omega_q(V(p,x)) = \omega_q(V^\#(p,x))$, for any $p \in \N$, by construction (step 4). Now, from Theorem \ref{Compression}, we have

\[\omega_q(V(p,x)) \geq \frac{1}{2} + \frac{\omega_q(V(2^p,x))}{2}, \]
and by $k$ applications of the theorem, we obtain 
\[ \omega_q(V(p,x)) \geq  \frac{\omega_q(V(\overbrace{2^{2^{{...}^{2^p}}}}^k,x))}{2^k} + \sum_{i = 1}^{k}\frac{1}{2^i}.\]
for every $k$. Taking the limit $k\to \infty$, we have $\omega_q(V(p,x)) = 1$ for all $p \in \N$. In particular $\omega_q(V(1,x)) = 1$.

Now suppose $x \notin L$. Then $\exists m \ldotp \forall n \ldotp R(x,m,n) = 0$. We prove that $\omega(V(1,x)) < 1$. Let $p$ be the smallest integer for which $R(x,\log^*(p),n)=0$ for every $n$. In other words, the Turing machine $R_{x,\log^*(p)}$ does not halt. Therefore by Theorem \ref{HALT} we have that $\omega_q(V(p,x)) \leq \omega_q(V_{\HALT,R_{x,\log^*(p)}}) \leq \frac{1}{2}$. 

If $p = 1$, we are done. Suppose $p > 1$. For all $k < p$, the game $V_{\HALT,R_{x,\log^*(k)}}$ never rejects since the Turing machine $R_{x,\log^*(k)}$ halts, by the minimality of $p$. Therefore $\omega_q(V(k,x)) = \omega_q(V^\#(k,x))$. So by recursively applying Theorem \ref{Compression}, we have that
$$\omega_q(V(p,x)) < 1 \implies \omega_q(V(1,x)) < 1.$$ 

Since $\omega_q(V(p,x)) \leq \omega_q(V_{\HALT,R_{x,\log^*(p)}}) \leq \frac{1}{2} $ then $\omega_q(V(1,x)) < 1$.
\end{proof}

\begin{corollary}
$\Pi_2^0\subseteq \MIPzero$.
\end{corollary} 
\begin{proof}
Let $L \in \Pi^0_2$ then $L = \{x\in \{0,1\}^\ast : \forall m\ldotp \exists n\ldotp R(x,m,n)=1 \}$. Let $U$ be the $1$-input $3$-prover verifier, that on input $x$ executes the verifier $V(1,x)$ where $x \in \{0,1\}^\ast$. By Claim \ref{bounded}, $\mathsf{TIME}(U(x)) = \mathsf{TIME}(V(1,x)) \leq \lambda(1+|x|)^\lambda$ and by Theorem \ref{mainThm}, $x \in L$ iff $\omega^\ast(U(x)) = 1$. Thus $U$ is an $\MIPzero$ protocol for the language $L$, and $L \in \MIPzero$.
\end{proof}

This concludes the proof of the main result of this paper.
 \section{$\MIPstar = \RE$ implies gap-preserving compression}
\label{sec:gap}

As mentioned in the introduction, the key to proving $\MIPstar = \RE$ in~\cite{ji2020mip} was establishing a gap-preserving compression theorem for non-local games. Here we observe that the reverse holds: $\MIPstar = \RE$ implies a gap-preserving compression theorem. 

\begin{theorem} If $\MIP^* = \RE$, then there exists a Turing machine $\COMPRESS$, with the following properties. Given as input a $k$-input $r$-prover verifier $V$, $\COMPRESS$ outputs a $k$-input $2$-prover verifier $V^\#$ in time polynomial in the description length of $V$, with the following properties:
\begin{enumerate}
    \item if $\omega_q(V(x_1, ...,x_k)) = 1$ then $\omega_q(V^\#(x_1, ...,x_k)) = 1$
    \item if $\omega_q(V(x_1, ...,x_k)) \leq \frac{1}{2}$ then $\omega_q(V^\#(x_1, ...,x_k)) \leq \frac{1}{2}$
    \item The runtime of the verifier $V^\#$ is polynomial in the description length of $V$ and its input.
\end{enumerate}

\end{theorem}
\begin{proof}
$\COMPRESS$ is the Turing machine that, when given input a verifier $V$, it returns the description of the verifier $V^\#$ from Figure \ref{fig:verifier}.

In the description of $V^\#$, we refer to the Turing machine $T_{V, (x_1, ...,x_k)}$. For every $k$-input $r$-prover verifier $V$ and $x_1,\ldots,x_k\in \{0,1\}^*$, $T_{V, (x_1, ...,x_k)}$ is the Turing machine that on empty tape enumerates over finite-dimensional quantum strategies for $V(x_1, ...,x_k)$ and only accepts if it finds a strategy that wins the game with probability greater than $\frac{1}{2}$. It does this via enumerating over $\epsilon$-nets (for $\epsilon = \frac{1}{4}$) for strategies of dimension $d$ for all $d \in \N$, as with the proof of Theorem \ref{upperbound}.

\begin{figure}[!htb]
  \centering
  \begin{gamespec}
    Input: $(x_1, ...,x_k)$, where $x_1,\ldots,x_k\in \{0,1\}^*$ \\
    Perform the following steps:
    
    \begin{enumerate}
        \item Compute $V_{\HALT,T_{V, (x_1, ...,x_k)}} = H(T_{V, (x_1, ...,x_k)})$ (where $H$ is from Theorem \ref{HALT}).
        \item Execute the interactive protocol specified by the verifier $V_{\HALT,T_{V, (x_1, ...,x_k)}}$ and accept if and only if the verifier accepts.
    \end{enumerate}

  \end{gamespec}
  \caption{Specification of the compressed verifier $V^\#$}
  \label{fig:verifier}
\end{figure}

By Theorem \ref{HALT}, if the Turing machine $T_{V, (x_1, ...,x_k)}$ halts then $$\omega_q(V_{\HALT,T_{V, (x_1, ...,x_k)}}) = 1,$$ otherwise $\omega_q(V_{\HALT,T_{V, (x_1, ...,x_k)}}) \leq \frac{1}{2}$. Also the runtime of $V_{\HALT,T_{V, (x_1, ...,x_k)}}$ is $p(|V|+ |x_1| +... + |x_n|)$, for some polynomial $p$.

Then if $\omega_q(V(x_1, ...,x_k)) = 1$ the Turing machine $T_{V, (x_1, ...,x_k)}$ finds a strategy that wins with probability greater than $ \frac{3}{4}$ and halts. Therefore $$\omega_q(V^\#(x_1, ...,x_k)) = \omega_q(V_{\HALT,T_{V, (x_1, ...,x_k)}}) = 1.$$

Otherwise, if $\omega_q(V(x_1, ...,x_k)) \leq \frac{1}{2}$ then there is no strategy that wins the game with probability $\frac{1}{2}$ and the Turing machine $T_{V, (x_1, ...,x_k)}$ never halts. Therefore $$\omega_q(V^\#(x_1, ...,x_k)) = \omega_q(V_{\HALT,T_{V, (x_1, ...,x_k)}}) \leq \frac{1}{2}.$$
\end{proof}

Note that in this gap-preserving compression theorem, the time complexity of the verifier $V^\#$ is polynomial in the \emph{description length} of $V$ and its input -- rather than the \emph{time complexity} of $V$.  
\bibliographystyle{alpha}
\bibliography{bibliography.bib}

\end{document}